\documentclass[letterpaper, 10 pt, conference]{ieeeconf}  

\IEEEoverridecommandlockouts                              
\usepackage{graphics} 
\usepackage{epsfig} 
\usepackage{amsmath} 
\usepackage{amssymb}  
\usepackage{accents}
\usepackage{amsfonts}
\usepackage{mathtools}
\usepackage{enumerate}
\usepackage{tikz}
\usepackage[font = small]{caption}
\usepackage{subcaption}
\usepackage[noadjust]{cite} 
\usepackage{verbatim}
\usepackage{color}

\usetikzlibrary{calc,positioning,shapes,shadows,arrows,fit}

\newtheorem{problem}{Problem}
\newtheorem{definition}{Definition}
\newtheorem{theorem}{Theorem}
\newtheorem{assumption}{Assumption}
\newtheorem{proposition}{Proposition}
\newtheorem{remark}{Remark}

\renewcommand{\baselinestretch}{0.98}

\title{\LARGE \bf
Robust Decentralized Abstractions for Multiple Mobile Manipulators
}

\author{Christos K. Verginis and Dimos V. Dimarogonas
\thanks{The authors are with the ACCESS Linnaeus Center, School of Electrical
Engineering, KTH Royal Institute of Technology, SE-100 44, Stockholm,
Sweden and with the KTH Center for Autonomous Systems. Email: {\tt\small \{cverginis, dimos\}@kth.se}. This work was supported by the H2020 ERC Starting Grant BUCOPHSYS, the Swedish Research Council (VR), the Knut och Alice Wallenberg Foundation, the European Union's Horizon 2020 Research and Innovation Programme under the Grant Agreement No. 644128 (AEROWORKS) and the EU H2020 Research and Innovation Programme under GA No. 731869 (Co4Robots).}
}

\begin{document}

\maketitle
\thispagestyle{empty}
\pagestyle{empty}

\begin{abstract}
This paper addresses the problem of decentralized abstractions for multiple mobile manipulators with $2$nd order dynamics. In particular, we propose decentralized controllers for the navigation of each agent among predefined regions of interest in the workspace, while guaranteeing at the same time inter-agent collision avoidance and connectivity maintenance for a subset of initially connected agents. In that way, the motion of the coupled multi-agent system is abstracted into multiple finite transition systems for each agent, which are then suitable for the application of temporal logic-based high level plans. The proposed methodology is decentralized, since each agent uses local information based on limited sensing capabilities. Finally, simulation studies verify the validity of the approach.
\end{abstract}

\section{INTRODUCTION} \label{sec:intro}

Multi-agent systems have gained a significant amount of attention in the last decades, due to the several advantages they yield with respect to single-agent setups. A recent direction in the multi-agent control and robotics field is the use of temporal logic languages for motion and/or action planning, since they provide a fully-automated correct-by-design controller synthesis approach for autonomous robots. Temporal logics, such as linear temporal logic (LTL), computation tree logic (CTL) or metric-interval temporal logic (MITL), provide formal high-level languages that can describe planning objectives more complex than the usual navigation techniques. The task specification is given as a temporal logic formula with respect to a discretized abstraction of the robot motion modeled as a finite transition system, and then, a high-level discrete plan is found by off-the-shelf model-checking algorithms, given the finite transition system and the task specification \cite{baier2008principles}. 

There exists a wide variety of works that employ temporal logic languages for multi-agent systems, e.g., \cite{Meng15,Belta2007,Bhatia2010,Bhatia2011,Cowlagi2016,Diaz2015,Fainekos2009,Filippidis2012}. The discretization o\texttt{}f a multi-agent system to an abstracted finite transition system necessitates the design of appropriate continuous-time controllers for the transition of the agents among the states of the transition system \cite{baier2008principles}. Most works in the related literature, however, including the aforementioned ones, either assume that there \textit{exist} such continuous controllers or adopt single- and double-integrator models, ignoring the actual dynamics of the agents. Discretized abstractions, including design of the discrete state space and/or continuous-time controllers, have been considered in \cite{Belta2005,belta2006controlling,reissig2011computing,tiwari2008abstractions,rungger2015state} for general systems and \cite{boskos2015decentralized, belta2004abstraction} for multi-agent systems. Another important issue concerning multi-agent abstractions that has not been addressed in the related literature is the collision avoidance between the robotic agents, which, unlike the unrealistic point-mass assumption that is considered in many works, can be more appropriately approximated by unions of rigid bodies. 

This work addresses the problem of decentralized abstractions for a team of mobile robotic manipulators, represented by a union of $3$D ellipsoids, among predefined regions of interest in the workspace. Mobile manipulators consist of a mobile base and a robotic arm, which makes them suitable for performing actions around a workspace (e.g., transportation of objects). In \cite{tanner2003nonholonomic} the authors consider the navigation of two mobile manipulators grasping an object, based on $3$D ellipsoids, whereas \cite{Loizou-RSS-14} deals with general-shape multi-agent navigation, both based on point-world transformations. Navigation of ellipsoidal agents while incorporating collision-avoidance properties was also studied in \cite{do2012coordination} for single-integrator dynamics, by transforming the ellipsoids to spheres. In our previous work \cite{verginis_ifac17}, we addressed a hybrid control framework for the navigation of mobile manipulators and their interaction with objects in a given workspace, proposing, however, a centralized solution.

In this work, we design robust continuous-time controllers for the navigation of the agents among the regions of interest. The proposed methodology is decentralized, since each agent uses only local information based on limited sensing capabilities. Moreover, we guarantee (i) inter-agent collision avoidance by introducing a novel transformation-free ellipsoid-based strategy, (ii) connectivity maintenance for a subset of the initially connected agents, which might be important for potential cooperative tasks, and (iii) kinematic singularity avoidance of the robotic agents. 

The rest of the paper is organized as follows. Section \ref{sec:prel} provides necessary notation and preliminary background and Section \ref{sec:PF} describes the tackled problem. The main results are given in Section \ref{sec:main results} and Section \ref{sec:simulations} presents simulation results. Finally, \ref{sec:concl} concludes the paper.

\section{PRELIMINARIES} \label{sec:prel}
\subsection{Notation}
The set of positive integers is denoted as $\mathbb{N}$ whereas the real and complex $n$-coordinate spaces, with $n\in\mathbb{N}$, are denoted as $\mathbb{R}^n$ and $\mathbb{C}^n$, respectively; $\mathbb{R}^n_{\geq 0}, \mathbb{R}^n_{> 0}, \mathbb{R}^n_{\leq 0}$ and $\mathbb{R}^n_{< 0}$ are the sets of real $n$-vectors with all elements nonnegative, positive, nonpositive, and negative, respectively. The notation $\|x\|$ is used for the Euclidean norm of a vector $x \in \mathbb{R}^n$. Given a a scalar function $y:\mathbb{R}^{n}\to\mathbb{R}$ and a vector $x\in\mathbb{R}^n$, we use the notation $\nabla_{x}y(x) = [\tfrac{\partial y}{\partial x_1}, \dots, \tfrac{\partial y}{\partial x_n}]^\top\in\mathbb{R}^n$. 
Define by $I_n \in \mathbb{R}^{n \times n}, 0_{m \times n} \in \mathbb{R}^{m \times n}$, the identity matrix and the $m \times n$ matrix with all entries zero, respectively; $\mathcal{B}_{c,r} = \{x \in \mathbb{R}^3: \|x-c\| \leq r\}$ is the $3$D sphere of center $c\in\mathbb{R}^{3}$ and radius $r \in \mathbb{R}_{\ge 0}$. The boundary of a set $A$ is denoted as $\partial A$ and its interior as $\accentset{\circ}{A} = A\backslash\partial A$. The vector connecting the origins of coordinate frames $\{A\}$ and $\{B$\} expressed in frame $\{C\}$ coordinates in $3$D space is denoted as $p^{\scriptscriptstyle C}_{{\scriptscriptstyle B/A}}\in{\mathbb{R}}^{3}$. 
For notational brevity, when a coordinate frame corresponds to an inertial frame of reference $\{I\}$, we will omit its explicit notation (e.g., $p_{\scriptscriptstyle B} = p^{\scriptscriptstyle I}_{\scriptscriptstyle B/I}, \omega_{\scriptscriptstyle B} = \omega^{\scriptscriptstyle I}_{\scriptscriptstyle B/I}$). All vector and matrix differentiations are derived with respect to an inertial frame $\{I\}$, unless otherwise stated.

\subsection{Cubic Equations and Ellipsoid Collision}


\begin{proposition} \label{prop:cubic}
Consider the cubic equation $f(\lambda) = c_3\lambda^3+c_2\lambda^2+c_1\lambda + c_0 = 0$ with $c_\ell\in\mathbb{R},\forall \ell\in\{0,\dots,3\}$ and roots $(\lambda_1,\lambda_2,\lambda_3)\in\mathbb{C}^3$, with $f(\lambda_1)=f(\lambda_2)=f(\lambda_3)=0$. Then, given its discriminant $\Delta = (c_3)^4\prod_{\substack{i\in\{1,2\}\\\substack{j\in\{i+1,\dots,3\} } }}(\lambda_i-\lambda_j)^2$, the following hold: 
\begin{enumerate} [(i)]
\item $\Delta = 0 \Leftrightarrow \exists i,j\in\{1,2,3\}$, with $i\neq j$, such that $\lambda_i=\lambda_j$, i.e., at least two roots are equal, 
\item $\Delta > 0 \Leftrightarrow \lambda_i\in\mathbb{R},\forall i\in\{1,2,3\}$, and $\lambda_i\neq\lambda_j, \forall i,j\in\{1,2,3\}$, with $i\neq j$, i.e., all roots are real and distinct.
\end{enumerate}
\end{proposition}


\begin{proposition} \cite{choi2009continuous} \label{prop:ellipsoids}
Consider two planar ellipsoids $\mathcal{A} = \{z\in\mathbb{R}^3 \text{ s.t. } z^\top A(t) z \leq 0\}$, $\mathcal{B} = \{z\in\mathbb{R}^3 \text{ s.t. } z^\top B(t) z \leq 0  \}$, with $z=[p^\top 1]^\top$ being the homogeneous coordinates of $p\in\mathbb{R}^2$, and $A, B:\mathbb{R}_{\geq 0}\to\mathbb{R}^{3\times3}$ terms that describe their motion in $2$D space.
Given their characteristic polynomial $f:\mathbb{R}\to\mathbb{R}$ with $f(\lambda) = \det(\lambda A - B)$, which has degree $3$, the following hold:
\begin{enumerate}[(i)]
\item  $\exists \lambda^*\in\mathbb{R}_{>0} \text{ s.t. } f(\lambda^*)=0$, i.e,
the polynomial $f(\lambda)$ always has one positive real root,
\item $\mathcal{A}\cap\mathcal{B} = \emptyset \Leftrightarrow \exists \lambda^*_1, \lambda^*_2 \in\mathbb{R}_{<0}$, with $\lambda^*_1\neq\lambda^*_2$, and $f(\lambda^*_1)=f(\lambda^*_2)=0$, i.e.,
$\mathcal{A}$ and $\mathcal{B}$ are disjoint if and only if the characteristic equation $f(\lambda) = 0$ has two distinct negative roots.
\item $\mathcal{A}\cap\mathcal{B} \neq \emptyset$ and $\accentset{\circ}{\mathcal{A}}\cap\accentset{\circ}{\mathcal{B}}=\emptyset \Leftrightarrow \exists \lambda^*_1, \lambda^*_2 \in\mathbb{R}_{<0}$, with $\lambda^*_1=\lambda^*_2$, and $f(\lambda^*_1)=f(\lambda^*_2)=0$, i.e.,
$\mathcal{A}$ and $\mathcal{B}$ touch externally if and only if the characteristic equation $f(\lambda) = 0$ has a negative double root.
\end{enumerate}
\end{proposition}


\section{PROBLEM FORMULATION} \label{sec:PF}

\begin{figure}[]
\vspace{0.4cm}
\centering
\includegraphics[trim = 0 0 0 0,scale=0.25]{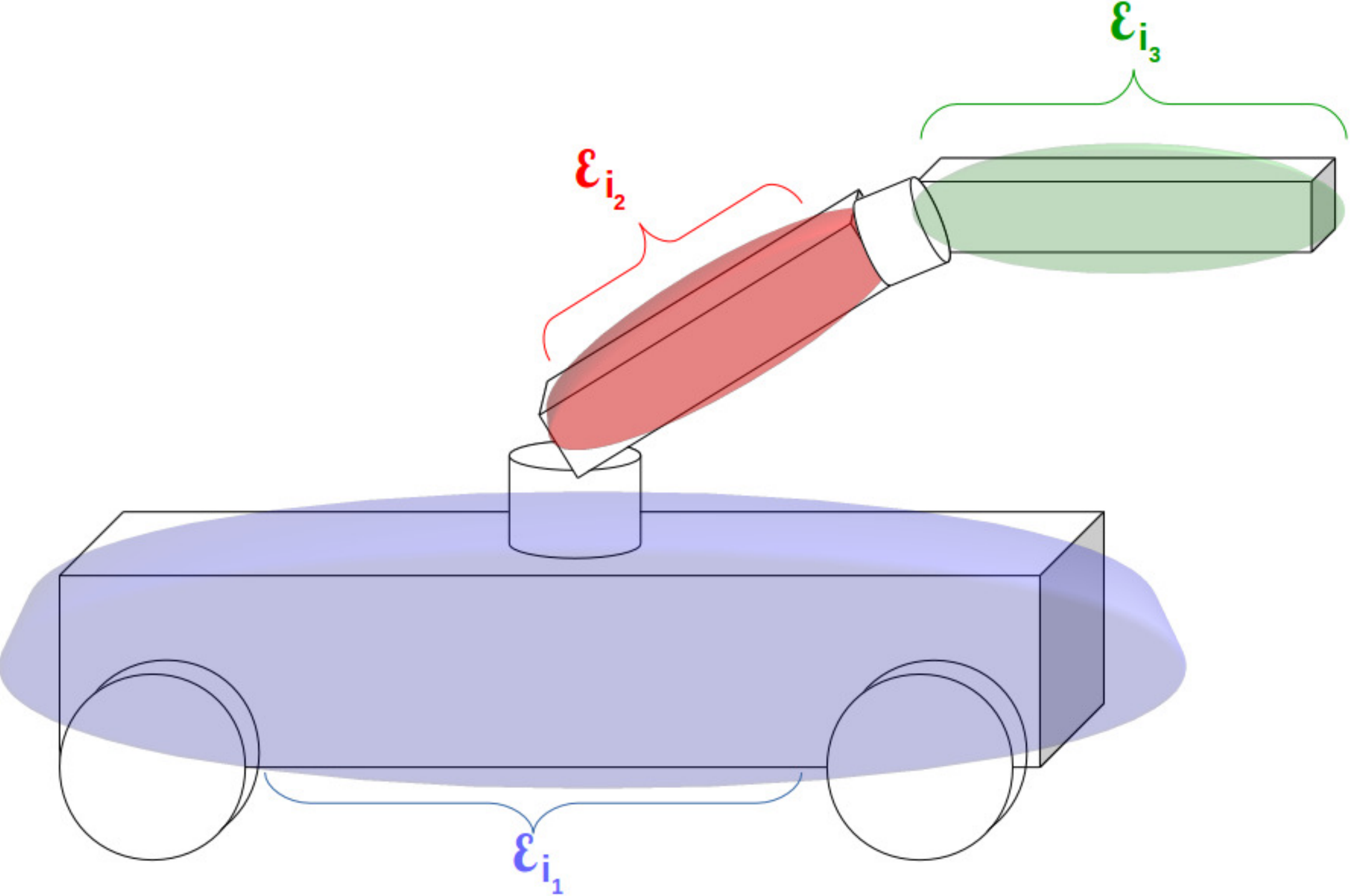}
\caption{An agent that consists of $\ell_i = 3$ rigid links.\label{fig:Agent}}
\end{figure}

Consider $N\in\mathbb{N}$ fully actuated agents with $\mathcal{V}\coloneqq \{1,\dots,N\}, N\geq 2$, composed by a robotic arm mounted on an
 omnidirectional mobile base, operating in a static workspace $\mathcal{W}$ that is bounded by a large sphere in $3$D space, i.e. $\mathcal{W} = \mathring{\mathcal{B}}_{p_0,r_0} = \{p\in\mathbb{R}^3 \text{ s.t. } \lVert p - p_0 \rVert < r_0\}$, where $p_0\in\mathbb{R}^3$ is the center of $\mathcal{W}$, and $r_0\in\mathbb{R}_{\geq 0}$ is its radius. Without loss of generality, we consider that $p_0 = 0_{3\times 1}$, corresponding to an inertial frame $\{I\}$.  Within $\mathcal{W}$ there exist $K$ disjoint spheres around points of interest, which are described by $\pi_k = \mathcal{B}_{p_k,r_k} = \{p\in\mathbb{R}^3 \text{ s.t. } \lVert p - p_k \rVert \leq r_k\}, k\in\mathcal{K}\coloneqq\{1,\dots,K\}$, where $p_k\in\mathbb{R}^3$ and $r_k\in\mathbb{R}_{>0}$ are the center and radius of the $k$th region, respectively. The regions of interest can be equivalently described by $\pi_k = \{z\in\mathbb{R}^4 \text{ s.t. } z^\top T_{\pi_k}z \leq 0 \}$, where $z=[p^\top, 1]^\top$ is the vector of homogeneous coordinates of $p\in\mathbb{R}^3$, and
\begin{equation} \label{eq:region matrix}
T_{\pi_k} = \begin{bmatrix}
I_3 & p_{k} \\ 0^\top_{3\times 1} & -r^2_k
\end{bmatrix}, \forall k\in\mathcal{K}.
\end{equation}
The dynamic model of each agent is given by the second-order Lagrangian dynamics:
\begin{equation}
M_i(q_i)\ddot{q}_i + N_i(q_i,\dot{q}_i)\dot{q}_i + g_i(q_i) + f_i(q_i, \dot{q}_i) = \tau_i, \label{eq:dynamics}
\end{equation}
$\forall i\in\mathcal{V}$, where $q_i\in\mathbb{R}^{n_i}$ is the vector of generalized  coordinates (e.g., pose of mobile base and joint coordinates of the arms), $M_i:\mathbb{R}^{n_i}\to\mathbb{R}^{n_i\times n_i}$ is the positive definite inertia matrix,  $N_i:\mathbb{R}^{n_i}\times\mathbb{R}^{n_i}\to\mathbb{R}^{n_i\times n_i}$ is the Coriolis matrix, $g_i:\mathbb{R}^{n_i}\to\mathbb{R}^{n_i}$ is the gravity vector, $f_i:\mathbb{R}^{n_i}\times\mathbb{R}^{n_i}\to\mathbb{R}^{n_i}$ is a term representing friction and modeling uncertainties and $\tau_i\in\mathbb{R}^{n_i}$ is the vector of joint torques, representing the control inputs. Without loss of generality, we assume that $n_i = n\in\mathbb{N},\forall i\in\mathcal{V}$. In addition, we denote as $\{B_i\}$ the frame of the mobile base of agent $i$ and $p_{\scriptscriptstyle B_i}:\mathbb{R}^{n}\to\mathbb{R}^3$ its inertial position. 
Moreover, the matrix $\dot{M}_i - 2N_i$ is skew-symmetric \cite{siciliano2008springer}, and we further make the following assumption:
\begin{assumption} \label{ass:f}
There exist positive constants $c_i$ such that $\lVert f_i(q_i,\dot{q}_i)\rVert\leq c_i\lVert q_i \rVert \lVert \dot{q}_i \rVert, \forall (q_i,\dot{q}_i)\in\mathbb{R}^{n}\times\mathbb{R}^{n}, i\in\mathcal{V}$.
\end{assumption}

We consider that each agent is composed by $\ell_i$ rigid links (see Fig. \ref{fig:Agent}) with $\mathcal{Q}_i = \{1,\dots,\ell_i\}$ the corresponding index set. Each link of agent $i$ is approximated by the ellipsoid set \cite{choi2009continuous} $\mathcal{E}_{i_m}(q_i) = \{z\in\mathbb{R}^4 \text{ s.t. } z^\top E_{i_m}(q_i)z \leq 0\}$; $z=[p^\top, 1]^\top$ is the homogeneous coordinates of $p\in\mathbb{R}^3$, and $E_{i_m}:\mathbb{R}^n\to\mathbb{R}^{4\times4}$ is defined as $E_{i_m}(q_i) = T^{-T}_{i_m}(q_i)\hat{E}_{i_m}T^{-1}_{i_m}(q_i)$, where $\hat{E}_{i_m} = \text{diag}\{a^{-2}_{i_m},b^{-2}_{i_m},c^{-2}_{i_m},-1\}$ corresponds to the positive lengths $a_{i_m},b_{i_m},c_{i_m}$ of the principal axes of the ellipsoid, and $T_{i_m}:\mathbb{R}^n\to\mathbb{R}^{4\times4}$ is the transformation matrix for the coordinate frame $\{i_m\}$ placed at the center of mass of the $m$-th link of agent $i$, aligned with the principal axes of $\mathcal{E}_{i_m}$:
\begin{equation}
T_{i_m}(q_i) = \begin{bmatrix}
R_{i_m}(q_i) & p_{i_m}(q_i) \\ 0^\top_{3\times 1} & 1
\end{bmatrix}, \notag
\end{equation}
with $R_{i_m}:\mathbb{R}^n\to\mathbb{R}^{3\times3}$ being the rotation matrix of the center of mass of the link, $\forall m\in\mathcal{Q}_i,i\in\mathcal{V}$.
For an ellipsoid $\mathcal{E}_{i_m}, i\in\mathcal{V},m\in\mathcal{Q}_i$, we denote as $\mathcal{E}^{xy}_{i_m},\mathcal{E}^{xz}_{i_m}, \mathcal{E}^{yz}_{i_m}$ its projections on the planes $x$-$y$, $x$-$z$ and $y$-$z$, respectively, with corresponding matrix terms $E^{xy}_{i_m},E^{xz}_{i_m}, E^{yz}_{i_m}$. Note that the following holds for two different ellipsoids $\mathcal{E}_{i_m}$ and $\mathcal{E}_{j_l}$:
\begin{align}
\mathcal{E}_{i_m}(q_i)\cap\mathcal{E}_{j_l}(q_j) \neq \emptyset \ &\land \
\accentset{\circ}{\mathcal{E}}_{i_m}(q_i)\cap\accentset{\circ}{\mathcal{E}}_{j_l}(q_j) = \emptyset
 \Leftrightarrow \notag \\
\mathcal{E}^s_{i_m}(q_i)\cap\mathcal{E}^s_{j_l}(q_j) \neq \emptyset \ &\land \ \accentset{\circ}{\mathcal{E}}^s_{i_m}(q_i)\cap\accentset{\circ}{\mathcal{E}}^s_{j_l}(q_j) = \emptyset, \notag
\end{align}
$\forall s\in\{xy,xz,yz\}$, i.e., in order for $\mathcal{E}_{i_m}, \mathcal{E}_{j_l}$ to collide (touch externally), all their projections on the three planes must also collide. Therefore, a sufficient condition for $\mathcal{E}_{i_m}$ and $\mathcal{E}_{j_l}$ not to collide is $\mathcal{E}^s_{i_m}(q_i)\cap\mathcal{E}^s_{j_l}(q_j)=\emptyset$, for some  $s\in\{xy,xz,yz\}$. In view of Proposition \ref{prop:ellipsoids}, that means that the characteristic equations $f^{s}_{i_m,j_l}(\lambda) \coloneqq \det(\lambda\mathcal{E}^s_{i_m}(q_i) - \mathcal{E}^s_{j_l}(q_j))=0$ must always have one positive real root and two negative distinct roots for at least one $s\in\{xy,xz,yz\}$. Hence, be denoting the discriminant of $f^{s}_{i_m,j_l}(\lambda) = 0$ as $\Delta^s_{i_m,j_l}$, Proposition \ref{prop:cubic} suggests that $\Delta^s_{i_m,j_l}$ must remain always positive for at least one $s\in\{xy,xz,yz\}$, since a collision would imply $\Delta^s_{i_m,j_l}=0$, $\forall s\in\{xy,xz,yz\}$. Therefore, by defining the function $\delta:\mathbb{R}\to\mathbb{R}_{\geq 0}$ as: 
\begin{equation} \label{eq:delta}
\delta(x) = \begin{cases}
\phi_\delta(x), & x > 0, \\
0, & x \leq 0,
\end{cases}
\end{equation}
where $\phi_\delta$ is an appropriate polynomial that ensures that $\delta(x)$ is twice continuously differentiable everywhere (e.g. $\phi_\delta(x)=x^3$), we can conclude that a sufficient condition for $\mathcal{E}_{i_m}$ and $\mathcal{E}_{j_l}$ not to collide is $\delta(\Delta^{xy}_{i_m,j_l}) + \delta(\Delta^{xz}_{i_m,j_l}) + \delta(\Delta^{yz}_{i_m,j_l}) > 0$, since a collision would result in $\Delta^s_{i_m,j_l} = 0 \Leftrightarrow \delta(\Delta^s_{i_m,j_l}) = 0, \forall s\in\{xy,xz,yz\}$. 

Next, we define the constant $\bar{d}_{\scriptscriptstyle B_i}$, which is the maximum distance of the base to a point in the agent's volume over all possible configurations, i.e. $\bar{d}_{\scriptscriptstyle B_i} = \sup_{q_i\in\mathbb{R}^n}\{ \lVert p_{\scriptscriptstyle B_i}(q_i) - p_i(q_i) \rVert \}, p_i\in\bigcup_{m\in\mathcal{Q}_i} \mathcal{E}_{i_m}(q_i)$. We also denote $\bar{d}_{\scriptscriptstyle B} = [\bar{d}_{\scriptscriptstyle B_1},\dots,\bar{d}_{\scriptscriptstyle B_N}]^\top\in\mathbb{R}_{\geq 0}^N$. 

Moreover, we consider that each agent has a sensor located at the center of its mobile base $p_{\scriptscriptstyle B_i}$ with a sensing radius $d_{\text{con}_i} \geq 2\max_{i\in\mathcal{V}}\{ \bar{d}_{\scriptscriptstyle B_i} \} + \varepsilon_d$, where $\varepsilon_d$ is an arbitrarily small positive constant. Hence, each agent has the sensing sphere $\mathcal{D}_i(q_i)=\{p\in\mathbb{R}^3 \text{ s.t. } \lVert p - p_{\scriptscriptstyle B_i}(q_i) \rVert \leq d_{\text{con}_i}  \}$ and its neighborhood set at each time instant is defined as $\mathcal{N}_i(q_i) = \{j\in\mathcal{V}\backslash\{i\} \text{ s.t. } \lVert p_{\scriptscriptstyle B_i}(q_i) - p_{\scriptscriptstyle B_j}(q_j)\rVert \leq d_{\text{con}_i} \}$. 


As mentioned in Section \ref{sec:intro}, we are interested in defining transition systems for the motion of the agents in the workspace in order to be able to assign complex high level goals through logic formulas. Moreover, since many applications necessitate the cooperation of the agents in order to execute some task (e.g. transport an object), we consider that a nonempty subset $\widetilde{\mathcal{N}}_i \subseteq \mathcal{N}_i(q_i(0)), i\in\mathcal{V}$, of the initial neighbors of the agents must stay connected through their motion in the workspace. In addition, it follows that the transition system of each agent must contain information regarding the current position of its neighbors. The problem in hand is equivalent to designing decentralized control laws $\tau_i,i\in\mathcal{V}$, for the appropriate transitions of the agents among the predefined regions of interest in the workspace.

Next, we provide the following necessary definitions.
\begin{definition} \label{def:in region}
An agent $i\in\mathcal{V}$ is in region $k \in\mathcal{K}$ at a configuration $q_i\in\mathbb{R}^n$, denoted as $\mathcal{A}_i(q_i)\in\pi_k$, if and only if $\lVert p_{i_m}(q_i) - p_k \rVert \leq r_k - \max\{\alpha_{i_m},\beta_{i_m},c_{i_m}\},\forall m\in\mathcal{Q}_i \Rightarrow \lVert p_{\scriptscriptstyle B_i}(q_i) - p_k \rVert \leq r_k	 - \bar{d}_{\scriptscriptstyle B_i}$.
\end{definition}

\begin{definition} 
Agents $i,j\in\mathcal{V}$, with $i\neq j$, are in \textit{collision-free} configurations $q_i,q_j\in\mathbb{R}^n$, denoted as $\mathcal{A}_i(q_i)\not \equiv\mathcal{A}_j(q_j)$, if and only if $\mathcal{E}_{i_m}(q_i)\cap\mathcal{E}_{j_l}(q_j)=\emptyset, \forall m\in\mathcal{Q}_i,l\in\mathcal{Q}_j$.
\end{definition}

Given the aforementioned discussion, we make the following assumptions regarding the agents and the validity of the workspace:
\begin{assumption} \label{ass:validity}
The regions of interest are
\begin{enumerate}[(i)]
\item  large enough such that all the robots can fit, i.e., given a specific $k\in\mathcal{K}$, there exist $q_i, i\in\mathcal{V}$, such that 
$\mathcal{A}_i(q_i)\in\pi_k$, $\forall i\in\mathcal{V}$, with $\mathcal{A}_i(q_i)\not\equiv\mathcal{A}_j(q_j)$, $\forall i,j\in\mathcal{V}$, with $i\neq j$. 
\item sufficiently far from each other and the obstacle workspace, i.e., 
\begin{align*}
& \lVert p_k - p_{k'} \rVert \geq \max\limits_{i\in\mathcal{V}}\{ 2\bar{d}_{\scriptscriptstyle B_i}\} + r_{k} + r_{k'} + \varepsilon_{p}, \notag \\
& r_0 - \| p_k \| \geq \max\limits_{i\in\mathcal{V}}\{ 2\bar{d}_{\scriptscriptstyle B_i}\},  
\end{align*}
$\forall k,k'\in\mathcal{K}, k\neq k'$, where $\varepsilon_p$ is an arbitrarily small positive constant.
\end{enumerate}
\end{assumption} 

Next, in order to proceed, we need the following definition.

\begin{definition}\label{def:agent transition}
Assume that $\mathcal{A}_i(q_i(t_0))\in\pi_k, i\in\mathcal{V}$, for some $t_0\in\mathbb{R}_{\geq 0},k\in\mathcal{K}$, with $\mathcal{A}_i(q_i(t_0))\not \equiv\mathcal{A}_j(q_j(t_0)), \forall j\in\mathcal{V}\backslash\{i\}$. There exists a transition for agent $i$ between $\pi_k$ and $\pi_{k'}, k'\in\mathcal{K}$, denoted as $(\pi_k,t_0)\xrightarrow{i}(\pi_{k'},t_f)$, if and only if there exists a finite time $t_f\geq t_0$, such that $\mathcal{A}_i(q_i(t_f))\in\pi_{k'}$ and $\mathcal{A}_i(q_i(t))\not \equiv\mathcal{A}_j(q_j(t))$, $\mathcal{E}_{i_m}(q_i(t))\cap\mathcal{E}_{i_\ell}(q_i(t))$,
 $\mathcal{E}_{i_m}(q_i(t))\cap\pi_{z} = \emptyset,\forall m,\ell\in\mathcal{Q}_i, m\neq \ell, j\in\mathcal{V}\backslash\{i\},z\in\mathcal{K}\backslash\{k,k'\}, t\in[t_0,t_f]$.
\end{definition}

Given the aforementioned definitions, the treated problem is the design of decentralized control laws for the transitions of the agents between two regions of interest in the workspace, while preventing collisions of the agents with each other, the workspace boundary, and the remaining regions of interest. More specifically, we aim to design a finite transition system for each agent of the form \cite{baier2008principles}
\begin{equation}
\mathcal{T}_i = (\Pi, \Pi_{i,0}, \xrightarrow{i}, \mathcal{AP}_i, \mathcal{L}_i, \mathcal{F}_i), \label{eq:TS}
\end{equation}
where $\Pi = \{\pi_1,\dots,\pi_K\}$ is the set of regions of interest that the agents can be at, according to Def. \ref{def:in region}, $\Pi_{i,0}\subseteq \Pi$ is a set of initial regions that each agent can start from, $\xrightarrow{i}\subset(\Pi\times\mathbb{R}_{\geq 0})^2$ is the transition relation of Def. \ref{def:agent transition}, $\mathcal{AP}_i$ is a set of given atomic propositions, represented as boolean variables, that hold in the regions of interest, $\mathcal{L}_i:\Pi\to2^{\mathcal{AP}_i}$ is a labeling function, and $\mathcal{F}_i:\Pi\to\Pi^{\lvert \widetilde{\mathcal{N}}_i \rvert}$ is a function that maps the region that agent $i$ occupies to the regions the initial neighbors $\widetilde{\mathcal{N}}_i$ of agent $i$ are at. Therefore, the treated problem is the design of bounded controllers $\tau_i$ for the establishment of the transitions $\xrightarrow{i}$.
Moreover, as discussed before, the control protocol should also guarantee the connectivity maintenance of a subset of the initial neighbors $\widetilde{\mathcal{N}	}_i,\forall i\in\mathcal{V}$. Another desired property important in applications involving robotic manipulators, is the nonsingularity of the Jacobian matrix $J_i:\mathbb{R}^n\to\mathbb{R}^{6\times n}$, that transforms the generalized coordinate rates of agent $i\in\mathcal{V}$ to generalized velocities \cite{siciliano2008springer}. That is, the set $\mathbb{S}_i = \{q_i\in\mathbb{R}^{n} \text{ s.t. } \det(J_i(q_i)[J_i(q_i)]^\top) = 0 \}$ should be avoided, $\forall i\in\mathcal{V}$.  
  Formally, we define the problem treated in this paper as follows:

\begin{problem} \label{Problem}
Consider $N$ mobile manipulators with dynamics \eqref{eq:dynamics} and $K$ regions of interest $\pi_k,k\in\mathcal{K}$, with $\dot{q}_i(t_0) < \infty, A_i(q_i(t_0))\in \pi_{k_i}, k_i\in\mathcal{K}, \forall i\in\mathcal{V}$ and $\mathcal{A}_i(q_i(t_0))\not\equiv\mathcal{A}_j(q_j(t_0)), \mathcal{E}_{i_m}(q_i(t_0))\cap\mathcal{E}_{i_\ell}(q_i(t_0)) =\emptyset, \forall i,j \in\mathcal{V},i\neq j, m,\ell\in\mathcal{Q}_i,m\neq\ell$. Given nonempty subsets of the initial edge sets $\widetilde{\mathcal{N}}_{i}\subseteq\mathcal{N}_i(q_i(0))\subseteq\mathcal{V}, \forall i\in\mathcal{V}$, the fact that $\det(J_i(q_i(t_0))[J_i(q_i(t_0))]^\top) \neq 0, \forall i\in\mathcal{V}$, as well as the indices $k'_{i}\in\mathcal{K},i\in\mathcal{V}$, such that $\lVert p_{k'_i} - p_{k'_j}\rVert + r_{k'_i} + r_{k'_j} \leq d_{\text{con}_i}, \forall j\in\widetilde{\mathcal{N}}_i,i\in\mathcal{V}$, design decentralized controllers $\tau_i$ such that, for all $i\in\mathcal{V}$: 
\begin{enumerate}
\item $(\pi_{k_i},t_0) \xrightarrow{i} (\pi_{k'_i},t_{f_i})$, for some $t_{f_i}\geq t_0$,
\item $r_0 - (\lVert p_{\scriptscriptstyle B_i}(t) \rVert  + \bar{d}_{\scriptscriptstyle B_i}) > 0,\forall t\in [t_0,t_{f_i}]$,
\item $j_i^*\in{\mathcal{N}}_i(q_i(t)), \forall j_i^*\in\widetilde{\mathcal{N}}_i, t\in [t_0,t_{f_i}]$,
\item $q_i(t)\in\mathbb{R}^n\backslash\mathbb{S}_i, \forall t\in [t_0,t_{f_i}]$.
\end{enumerate} 
\end{problem}
The aforementioned specifications concern 1) the agent transitions according to Def. \ref{def:agent transition}, 2) the confinement of the agents in $\mathcal{W}$, 3) the connectivity maintenance between a subset of initially connected agents and 4) the agent singularity avoidance. Moreover, the fact that the initial edge sets $\widetilde{\mathcal{N}}_i$ are nonempty implies that the sensing radius of each agent $i$ covers the regions $\pi_{k_j}$ of the agents in the neighboring set $\widetilde{\mathcal{N}}_i$. Similarly, the condition $\lVert p_{k'_i} - p_{k'_j}\rVert + r_{k'_i} + r_{k'_j} \leq d_{\text{con}_i}, \forall j\in\widetilde{\mathcal{N}}_i$, is a feasibility condition for the goal regions, since otherwise it would be impossible for two initially connected agents to stay connected. Intuitively, the sensing radii $d_{\text{con}_i}$ should be large enough to allow transitions of the multi-agent system to the entire workspace.


\section{MAIN RESULTS}  \label{sec:main results}

\subsection{Continuous Control Design} \label{subsec:Continuous design}

To solve Problem \ref{Problem}, we denote as $\varphi_i:\mathbb{R}^{Nn}\to\mathbb{R}_{\geq 0}$ a \emph{decentralized potential function},
with the following properties: 
\begin{enumerate}[(i)]
	\item The function $\varphi_i(q)$ is not defined, i.e., $\varphi_i(q) = \infty$, $\forall i\in\mathcal{V}$, when a collision or a connectivity break occurs,
	\item The critical points of $\varphi_i$ where the vector field $\nabla_{q_i}\varphi_i(q)$ vanishes, i.e., the points where $\nabla_{q_i}\varphi_i(q) = 0$, consist of the goal configurations and a set of configurations whose region of attraction (by following the negated vector field curves) is a set of measure zero.	
	\item It holds that $\nabla_{q_i}\varphi_i(q) + \sum_{j\in\mathcal{N}_i(q_i)}\nabla_{q_i}\varphi_j(q) = 0$ $\Leftrightarrow$ $\nabla_{q_i}\varphi_i(q) = 0$ and $\sum_{j\in\mathcal{N}_i(q_i)}\nabla_{q_i}\varphi_j(q) = 0$, $\forall i\in\mathcal{N},q\in \mathbb{R}^{Nn} $.
\end{enumerate}
More specifically, $\varphi_i(q)$ is a function of two main terms, a  \emph{goal function} $\gamma_i:\mathbb{R}^{n}\to\mathbb{R}_{\geq 0}$, that should vanish when $\mathcal{A}_i(q_i)\in\pi_{k'_i}$, and an \emph{obstacle function}, 
$\beta_i:\mathbb{R}^{n}\to\mathbb{R}_{\geq 0}$ is a bounded that encodes inter-agent collisions, collisions between the agents and the obstacle boundary/undesired regions of interest, connectivity losses between initially connected agents and singularities of the Jacobian matrix $J_i(q_i)$; 
$\beta_i$ vanishes when one or more of the above situation occurs.
Next, we provide an analytic construction of the goal and obstacle terms. However, the construction of the function $\varphi_i$ is out of the scope of this paper. Examples can be found in \cite{dimarogonas2007decentralized}\footnote{In that case, we could choose $\varphi_i = \tfrac{1}{1-\phi_i}$, where $\phi_i$ is the proposed function of \cite{dimarogonas2007decentralized}} and \cite{panagou2017distributed}. 

\subsubsection{$\gamma_i$ - Goal Function}

Function $\gamma_i$ encodes the control objective of agent $i$, i.e., reach the region of interest $\pi_{k'_i}$. Hence, we define  $\gamma_i:\mathbb{R}^n\to\mathbb{R}_{\geq 0}$ as
\begin{equation} \label{eq:gamma_i}
\gamma_i(q_i) = \lVert q_i - q_{k'_i} \rVert^2,
\end{equation}
where $q_{k'_i}$ is a configuration such that $r_{k'} - \|p_{\scriptscriptstyle B_i}(q_{k'_i}) - p_{k'_i} \| \leq \bar{d}_{\scriptscriptstyle B_i} - \varepsilon$, for an arbitrarily small positive constant $\varepsilon$, which implies $\mathcal{A}_i(q_{k'_i})\in \pi_{k'_i}$, $\forall i\in\mathcal{V}$. In case that multiple agents have the same target, i.e., there exists at least one $j\in\mathcal{V}\backslash\{i\}$ such that $\pi_{k'_j} = \pi_{k'_i}$, then we assume that $\mathcal{A}_i(q_{k'_i})\not\equiv \mathcal{A}_j(q_{k'_j})$.

\subsubsection{$\beta_i$ - Collision/Connectivity/Singularity Function} \label{sec:beta_gamma_definitions}

The function $\beta_i$ encodes all inter-agent collisions, collisions with the boundary of the workspace and the undesired regions of interest, connectivity between initially connected agents and singularities of the Jacobian matrix $J_i(q_i),\forall i\in\mathcal{V}$.

Consider the function $\Delta_{i_m,j_l}:\mathbb{R}^{2n}\to\mathbb{R}_{\geq 0}$, with $\Delta_{i_m,j_l}(q_i,q_j) = \delta(\Delta^{xy}_{i_m,j_l}(q_i,q_j)) + \delta(\Delta^{xz}_{i_m,j_l}(q_i,q_j)) + \delta(\Delta^{yz}_{i_m,j_l}(q_i,q_j))$, where 
$\Delta^s_{i_m,j_l}:\mathbb{R}^{2n}\to\mathbb{R}_{\geq 0}$ is the discriminant of the cubic equation $\det\{\lambda E^s_{i_m}(q_i) - E^s_{j_l}(q_j)\}=0, \forall s\in\{xy,xz,yz\}$, for two given ellipsoids $\mathcal{E}_{i_m}$ and $\mathcal{E}_{j_l}, m\in\mathcal{Q}_i,l\in\mathcal{Q}_j,i,j,\in\mathcal{V}$, and $\delta$ as defined in \eqref{eq:delta}. 
As discussed in Section \ref{sec:PF}, a sufficient condition for the ellipsoids $\mathcal{E}_{i_m}$ and $\mathcal{E}_{j_l}$ not to collide, is $\Delta_{i_m,j_l}(q_i(t),q_j(t)) > 0, \forall t\in\mathbb{R}_{\geq 0}$.

Additionally, we define the greatest lower bound of the $\Delta_{i_m,j_l}$ when the point $p_{j_l}$ is on the boundary of the sensing radius $\partial D_i(q_i)$ of agent $i$, as $\widetilde{{\Delta}}_{i_m,j_l} = \inf_{(q_i,q_j)\in\mathbb{R}^{2n}}\{\Delta_{i_m,j_l}(q_i,q_j)\} \text{ s.t. } \lVert p_{\scriptscriptstyle B_i}(q_i) - p_{j_l}(q_j) \rVert = d_{\text{con}_i}, \forall m\in\mathcal{Q}_i,l\in\mathcal{Q}_j,i,j\in\mathcal{V}$. Since $d_{\text{con}_i} > 2\max_{i\in\mathcal{V}}\{\bar{d}_{\scriptscriptstyle B_i}\} + \varepsilon_d$, it follows that there exists a positive constant $\varepsilon_\Delta$ such that $\widetilde{\Delta}_{i_m,j_l} \geq \varepsilon_\Delta > 0,\forall m\in\mathcal{Q}_i,l\in\mathcal{Q}_j,i,j\in\mathcal{V}, i\neq j$.

Moreover, we define the function $\Delta_{i_m,\pi_k}:\mathbb{R}^n\to\mathbb{R}_{\geq 0}$, with 
$\Delta_{i_m,\pi_k}(q_i) = \delta(\Delta^{xy}_{i_m,\pi_k}(q_i))+\delta(\Delta^{xz}_{i_m,\pi_k}(q_i))+\delta(\Delta^{yz}_{i_m,\pi_k}(q_i))$, where $\Delta^s_{i_m,\pi_k}:\mathbb{R}^n\to\mathbb{R}$ is the discriminant of the cubic equation $\det(\lambda E^s_{i_m}(q_i) - T^s_{\pi_k})$, with $T^s_{\pi_k}$ the projected version of $T_{\pi_k}$ in \eqref{eq:region matrix}, $s\in\{xy,xz,yz\}$, and $\delta$ as given in \eqref{eq:delta}.
A sufficient condition for $\mathcal{E}_{i_m}$ and region $\pi_k, k\in\mathcal{K}$ not to collide is $\Delta_{i_m,\pi_k}(q_i(t))>0, \forall t\in\mathbb{R}_{\geq 0}, m\in\mathcal{Q}_i,i\in\mathcal{V}$.

We further define the function $\eta_{ij,c}:\mathbb{R}^{n}\times\mathbb{R}^{n}\to\mathbb{R}$, with $\eta_{ij,c}(q_i,q_j) = d^2_{\text{con}_i}-\lVert p_{\scriptscriptstyle B_i}(q_i) - p_{\scriptscriptstyle B_j}(q_j)\rVert^2$, and the distance functions $\beta_{i_m,j_l}:\mathbb{R}_{\geq 0}\to\mathbb{R}_{\geq 0}, \beta_{ij,c}:\mathbb{R}\to\mathbb{R}_{\geq 0}, \beta_{iw}:\mathbb{R}_{\geq 0}\to\mathbb{R}$ as 
\begin{align}
\beta_{i_m,j_l}(\Delta_{i_m,j_l}) &=
\begin{cases}
\phi_{i,a}(\Delta_{i_m,j_l}), &  0 \leq \Delta_{i_m,j_l} < \bar{\Delta}_{i_m,j_l}, \\
\bar{\Delta}_{i_m,j_l}, &  \bar{\Delta}_{i_m,j_l} \leq \Delta_{i_m,j_l}, \\
\end{cases} \notag \\ 
\beta_{ij,c}(\eta_{ij,c}) &= 
\begin{cases}
0, &  \eta_{ij,c} < 0,\\
\phi_{i,c}(\eta_{ij,c}), &   0 \le \eta_{ij,c} < d^2_{\text{con}_i}, \\
d^2_{\text{con}_i}, &  d^2_{\text{con}_i} \le \eta_{ij,c},  
\end{cases} \notag \\
\beta_{iw}(\lVert p_{\scriptscriptstyle B_i} \rVert^2) &=  (r_w - \bar{d}_{\scriptscriptstyle B_i})^2 - \lVert p_{\scriptscriptstyle B_i} \rVert^2, \notag 
\end{align}
where $\bar{\Delta}_{i_m,j_l}$ is a constant satisfying $0 < \bar{\Delta}_{i_m,j_l} \leq \widetilde{\Delta}_{i_m,j_l}, \forall m\in\mathcal{Q}_i,l\in\mathcal{Q}_j,i,j\in\mathcal{V},i\neq j$, and $\phi_{i,a}, \phi_{i,c}$ are \textit{strictly increasing} polynomials appropriately selected to guarantee that the functions $\beta_{i_m,j_l}$, and $\beta_{ij,c}$, respectively, are twice continuously differentiable everywhere, with $\phi_{i,a}(0) = \phi_{i,c}(0) = 0, \forall i\in\mathcal{V}$. 
Note that the functions defined above use only local information in the sensing range $d_{\text{con}_i}$ of agent $i$. The function $\beta_{i_m,j_l}$ becomes zero when ellipsoid $\mathcal{E}_{i_m}$ collides with ellipsoid $\mathcal{E}_{j_l}$, whereas $\beta_{ij,c}$ becomes zero when agent $i$ loses connectivity with agent $j$. Similarly, $\beta_{iw}$ encodes the collision of agent $i$ with the workspace boundary. 

Finally, we choose the function $\beta_i:\mathbb{R}^{Nn}\to\mathbb{R}_{\geq 0}$ as 
\begin{align}
\beta_i(q) =& (\det(J_i(q_i)[J_i(q_i)]^\top))^2\beta_{iw}(\lVert p_{\scriptscriptstyle B_i} \rVert^2)\prod\limits_{j\in\widetilde{\mathcal{N}}_i}\beta_{ij,c}(\eta_{ij,c}) \notag\\ &\prod\limits_{(m,j,l)\in\widetilde{T}}\beta_{i_m,j_l}(\Delta_{i_m,j_l})\prod\limits_{ (m,k)\in\widetilde{L}} \Delta_{i_m,\pi_k}(q_i), \label{eq:betas}
\end{align}
$\forall i\in\mathcal{V}$, where $\widetilde{T} = \mathcal{Q}_i\times\mathcal{V}\times\mathcal{Q}_j, \widetilde{L} = \mathcal{Q}_i\times(\mathcal{K}\backslash\{k_i,k'_i\})$, and we have omitted the dependence on $q$ for brevity.  Note that we have included the term $(\det(J_iJ^\top_i))^2$ to also account for singularities of $J_i, \forall i\in\mathcal{V}$ and the term $\prod_{(m,j,l)\in\widetilde{T}}\beta_{i_m,j_l}(\Delta_{i_m,j_l})$ takes into account also the collisions between the ellipsoidal rigid bodies of agent $i$.

With the introduced notation, the properties of the functions $\varphi_i$ are: 
\begin{enumerate}[(i)]
	\item $\beta_i(q)\to 0 \Leftrightarrow (\varphi_i(q) \to \infty), \forall i\in\mathcal{V}$,
	\item $ \nabla_{q_i}\varphi_i(q)|_{q_i=q^\star_i} = 0,  \forall q^\star_i\in\mathbb{R}^n \text{ s.t. } \gamma_i(q^\star_i) = 0$ and the regions of attraction of the points $\{q \in\mathbb{R}^{Nn}: \nabla_{q_i}\varphi_i(q)|_{q_i=\widetilde{q}_i} = 0, \gamma_i(\widetilde{q}_i) \neq 0 \}, i\in\mathcal{V}$, are sets of measure zero.
\end{enumerate}

By further denoting $\mathbb{D}_i = \{q\in\mathbb{R}^{Nn}: \beta_i(q) > 0 \}$, we are ready to state the main theorem, that summarizes the main results of this work.

\begin{theorem}
Under the Assumptions \ref{ass:f}-\ref{ass:validity}, the decentralized control laws $\tau_i: \mathbb{D}_i\times\mathbb{R}^{n}\to \mathbb{R}^n$, with
\begin{align}
& \tau_i(q,\dot{q}_i) = g_i(q_i) - \nabla_{q_i}\varphi_i(q) - \sum\limits_{j\in\mathcal{N}_i(q_i)}\nabla_{q_i}\varphi_j(q) \notag \\
&\hspace{15mm} - \hat{c}_i(q_i,\dot{q}_i)\lVert q_i \rVert \dot{q}_i - \lambda_{i} \dot{q}_i, \label{eq:control law}
\end{align}
$\forall i\in\mathcal{V}$, along with the adaptation laws $\dot{\hat{c}}_i:\mathbb{R}^n\times\mathbb{R}^n\to\mathbb{R}$:
\begin{equation}
\dot{\hat{c}}_i(q_i,\dot{q}_i) = \sigma_i\lVert \dot{q}_i\rVert^2\lVert q_i\rVert, \label{eq:adaptation law} 
\end{equation}
with $\hat{c}_i(q_i(t_0),\dot{q}_i(t_0)) < \infty,\sigma_i\in\mathbb{R}_{\geq 0}$ , $\forall i\in\mathcal{V}$, guarantee the transitions $(\pi_{k_i},t_0)\xrightarrow{i}(\pi_{k'_i},t_{f_i})$ for finite $t_{f_i},i\in\mathcal{V}$ for almost all initial conditions, while ensuring $\beta_i > 0,\forall i\in\mathcal{V}$, as well as the boundedness of all closed loop signals, providing, therefore, a solution to Problem \ref{Problem}.
\end{theorem}
\begin{proof}
The closed loop system of \eqref{eq:dynamics} is written as:
\small
\begin{align}
M_i(q_i)\ddot{q}_i + N_i(q_i,\dot{q}_i)\dot{q}_i + f_i(q_i,\dot{q}_i) = -\nabla_{q_i}\varphi_i(q_i) - \lambda_i\dot{q}_i \notag \\
- \hat{c}(q_i,\dot{q}_i)\lVert q_i \rVert\dot{q}_i  - \sum\limits_{j\in\mathcal{N}_i(q_i)}\nabla_{q_i}\varphi_j(q), \label{eq:closed loop dynamics}
\end{align}
\normalsize
$\forall i\in\mathcal{V}$.
Due to Assumption \ref{ass:validity}, the domain where the functions $\varphi_i(q)$ are well-defined (i.e., where $\beta_i > 0$) is connected. Hence,
consider the Lyapunov-like function $V:\mathbb{R}^{N}\times\mathbb{R}^{Nn}\times\mathbb{R}^N\times\mathbb{D}_1\times\dots\times\mathbb{D}_N\to\mathbb{R}_{\geq 0}$, with
\begin{align}
V(\varphi, \dot{q}, \widetilde{c},q) =& \sum\limits_{i\in\mathcal{V}} \varphi_i(q) + \frac{1}{2}[\dot{q}^\top_iM_i(q_i)\dot{q}_i + \frac{1}{\sigma_i}\widetilde{c}_i^2] \notag
\end{align}
where $\varphi$ and $\widetilde{c}$ are the stack vectors containing all $\varphi_i$ and $\widetilde{c}_i$, respectively, $i\in\mathcal{V}$, and $\widetilde{c}_i:\mathbb{R}^{n}\times\mathbb{R}^{n}\to\mathbb{R}$, with $\widetilde{c}_i(q_i,\dot{q}_i) = \hat{c}_i(q_i,\dot{q}_i) - c_i, \forall i\in\mathcal{V}$. Note that, since there are no collision or singularities at $t_0$, the functions $\beta_i(q), i\in\mathcal{V}$, are strictly positive at $t_0$ which implies the boundedness of $V$ at $t_0$. Therefore, since $\dot{q}_i(t_0)<\infty$ and $\hat{c}_i(t_0)<\infty, \forall i\in\mathcal{V}$, there exists a positive and finite constant $M<\infty$ such that $V_0\coloneqq V(\varphi(q(t_0)), \dot{q}(t_0), \tilde{c}(q(t_0),q(t_0)) \leq M$.

By differentiating $V$, substituting the dynamics \eqref{eq:dynamics}, employing the skew symmetry of $\dot{M}_i - 2N_i$ as well as the property $\sum_{i\in\mathcal{V}} ( [\nabla_{q_i}\varphi_i(q)]^\top\dot{q}_i + \sum_{j\in\mathcal{N}_i(q_i)} [\nabla_{q_j}\varphi_i(q)]^\top\dot{q}_j ) = \sum_{i\in\mathcal{V}}( [\nabla_{q_i}\varphi_i(q) ]^\top +  \sum_{j\in\mathcal{N}_i(q_i)} [\nabla_{q_i}\varphi_j(q)]^\top) \dot{q}_i$, we obtain
\begin{align}
\dot{V} =& \sum\limits_{i\in\mathcal{V}}\Bigg\{\dot{q}^\top_i \Bigg( \nabla_{q_i}\varphi_i(q) + \sum\limits_{j\in\mathcal{N}_i(q_i)} \nabla_{q_i}\varphi_j(q) + \tau_i - g_i(q_i)\Bigg)  \notag \\ 
&-\dot{q}_i^\top f_i(q_i,\dot{q}_i) + \frac{1}{\sigma_i}\widetilde{c}_i\dot{\hat{c}}_i  \Bigg\},
\end{align}
which, by substituting the control and adaptation laws \eqref{eq:control law}, \eqref{eq:adaptation law} becomes:
\begin{align}
\dot{V} =& \sum\limits_{i\in\mathcal{V}}\{ -\lambda_i\lVert \dot{q}_i \rVert^2 - \hat{c}_i\lVert \dot{q}_i \rVert^2 \lVert q_i \rVert  - \dot{q}_i^\top f_i(q_i,\dot{q}_i) \notag \\
& + \widetilde{c}_i\lVert \dot{q}_i \rVert^2 \lVert q_i \rVert  \notag\\ 
\leq &  \sum\limits_{i\in\mathcal{V}}\{ -\lambda_i\lVert \dot{q}_i \rVert^2 - (\hat{c}_i-c_i-\widetilde{c}_i)\lVert \dot{q}_i \rVert^2 \lVert q_i \rVert   
\end{align}
where we have used the property $\lVert f_i(q_i,\dot{q}_i) \rVert \leq c_i\lVert q_i \rVert \lVert \dot{q}_i\rVert$. Since $\widetilde{c}_i = \hat{c}_i - c_i$, we obtain $\dot{V} \leq - \sum_{i\in\mathcal{V}}\lambda_i\lVert \dot{q}_i\rVert^2$,
which implies that $V$ is non-increasing along the trajectories of the closed loop system. Hence, we conclude that $V(t)\leq V_0 \leq M$, as well as the boundedness of $\widetilde{c}_i, \varphi_i, \dot{q}_i$ and hence of $\hat{c}_i, \forall i\in\mathcal{V}, t\geq t_0$. Therefore, we conclude that $\beta_i(q(t)) > 0, \forall t\geq t_0, i\in\mathcal{V}$. 
Hence, inter-agent collisions, collision with the undesired regions and the obstacle boundary, connectivity losses between the subsets of the initially connected agents and singularity configurations are avoided.
Moreover, by invoking LaSalle's Invariance Principle, the system converges to the largest invariant set contained in
\begin{equation}
S = \{(q,\dot{q})\in \mathbb{D}_1\times\dots\times\mathbb{D}_N\times\mathbb{R}^{Nn}\text{ s.t. } \dot{q} = 0_{Nn\times 1}\}. \label{eq:la salle0}
\end{equation}
For $S$ to be invariant, we require that $\ddot{q}_i = 0_{n\times 1}, \forall i\in\mathcal{V}$, and thus we conclude for the closed loop system \eqref{eq:closed loop dynamics} that $\nabla_{q_i}\varphi_i(q) = 0_{n\times 1}, \forall i\in\mathcal{V}$,
since $\| f_i(q_i,0_{n\times1}) \| \leq 0, \forall q_i\in\mathbb{R}^{n}$, in view of Assumption \ref{ass:f}. Therefore, by invoking the properties of $\varphi_i(q)$, each agent $i\in\mathcal{V}$ will converge to a critical point of $\varphi_i$, i.e., all the configurations where $\nabla_{q_i}\varphi_i(q) = 0_{n\times1}, \forall i\in\mathcal{V}$. However, due to properties of $\varphi_i(q)$, the initial conditions that lead to configurations $\widetilde{q}_i$ such that  $\nabla_{q_i}\varphi_i(q)|_{q_i=\widetilde{q}_i} = 0_{n\times1}$ and $\gamma_i(\widetilde{q}_i) \neq 0$ are sets of measure zero in the configuration space \cite{Koditchek92}. Hence, the agents will converge to the configurations where $\gamma_i(q_i) = 0$ from almost all initial conditions, i.e., $\lim\limits_{t\to\infty}\gamma_i(q_i(t)) = 0$. Therefore, since $r_{k'} - \|p_{\scriptscriptstyle B_i}(q_{k'_i}) - p_{k'_i} \| \leq \bar{d}_{\scriptscriptstyle B_i} - \varepsilon$, it can be concluded that there exists a finite time instance $t_{f_i}$ such that $\mathcal{A}_i(q_i(t_{f_i}))\in\pi_{k'}$, $\forall i\in\mathcal{V}$ and hence, each agent $i$ will be at its goal region $p_{k'_i}$ at time $t_{f_i}, \forall i\in\mathcal{V}$. In addition, the boundedness of $q_i,\dot{q}_i$ implies the boundedness of the adaptation laws $\dot{\hat{c}}_i, \forall i\in\mathcal{V}$. Hence, the control laws \eqref{eq:control law} are also bounded.
\begin{figure*}
\vspace{0.4cm}
\begin{subfigure}[t]{0.32\textwidth}	
	\includegraphics[scale = 0.31]{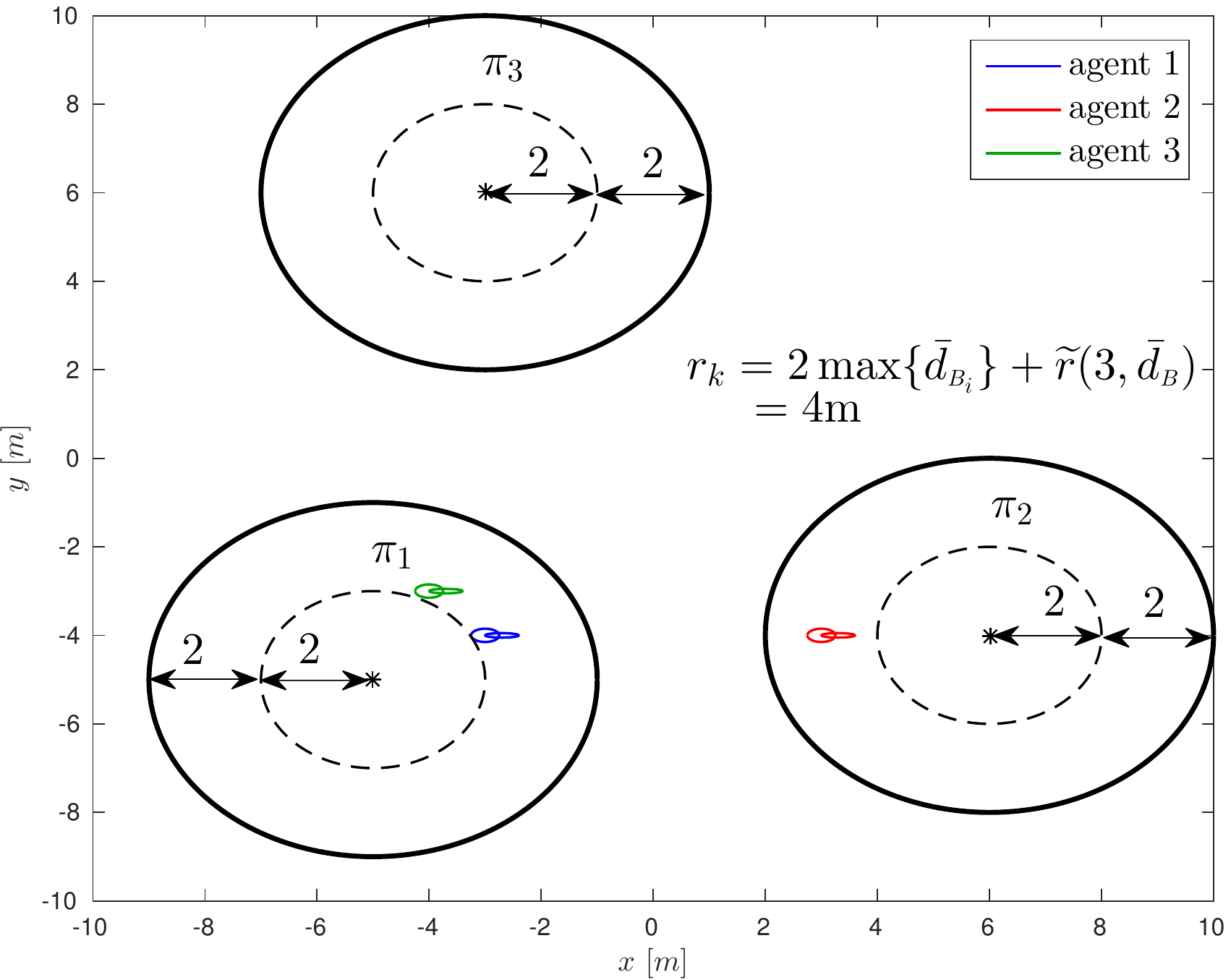}	
	\caption{\label{fig:navigation initial}}
\end{subfigure} %
\begin{subfigure}[t]{0.32\textwidth}
	\includegraphics[trim = -0.65cm 0 0 0,scale = 0.31]{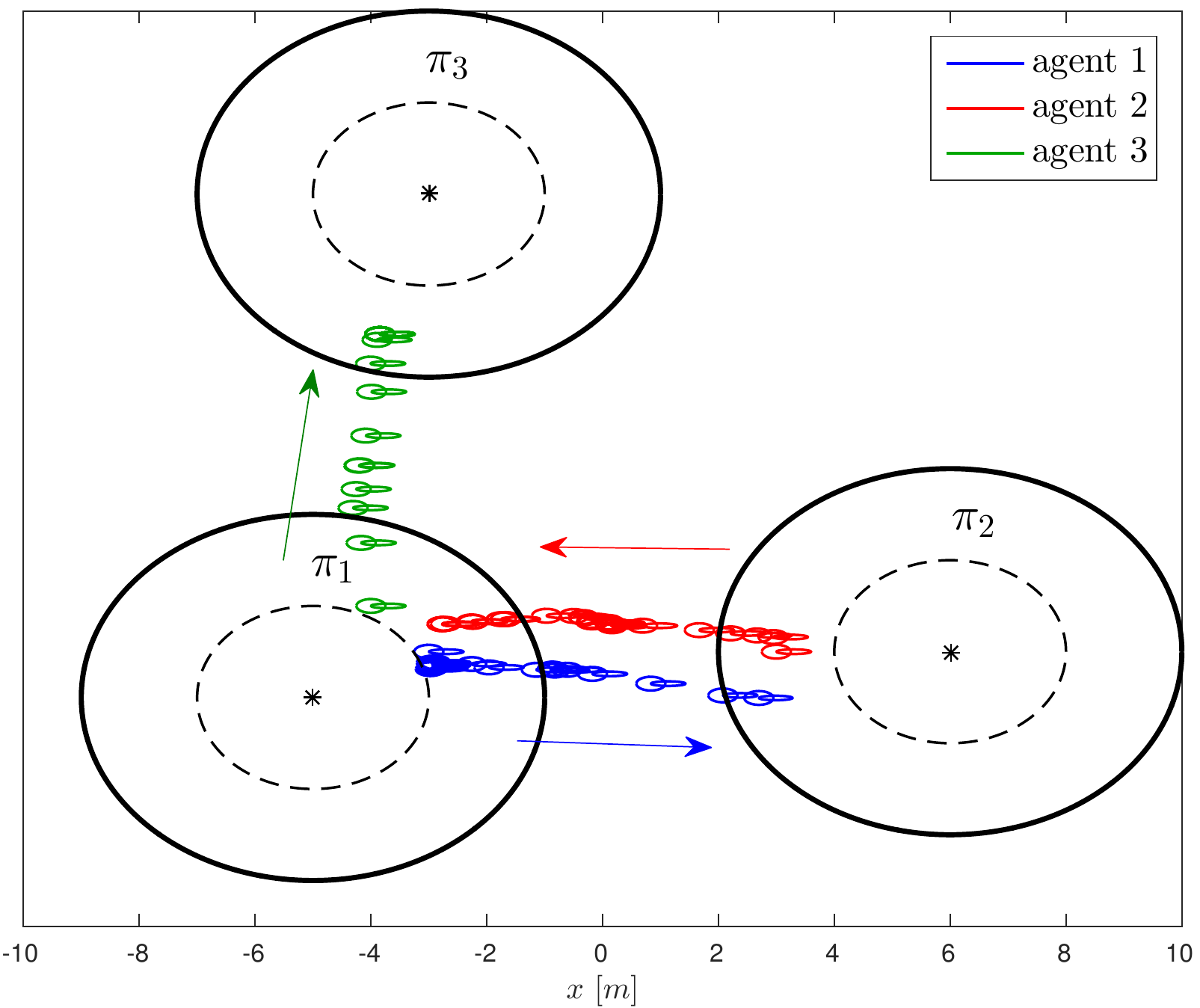}	
	\caption{\label{fig:navigation final}}
\end{subfigure} %
\begin{subfigure}[t]{0.32\textwidth}	
	\includegraphics[scale = 0.31]{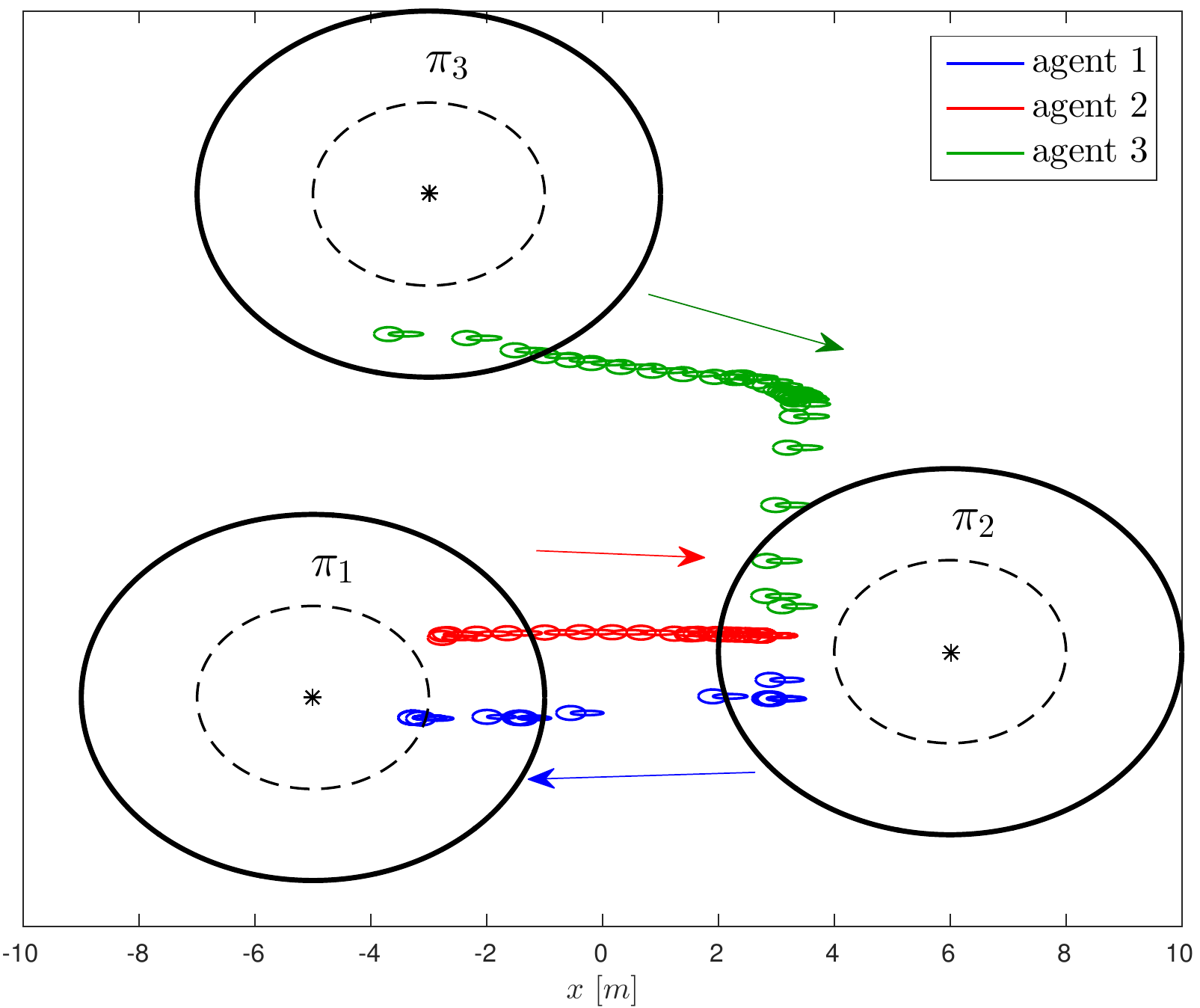}	
	\caption{\label{fig:navigation final2}}
\end{subfigure} %
(a): The initial position of the agents in the workspace of the simulation example. (b): The first transition of the agents in the workspace. Agent $1$ transits from $\pi_1$ to $\pi_2$, agent $2$ from $\pi_2$ to $\pi_1$, and agent $3$ from $\pi_{1}$ to $\pi_3$. (c): The second transition of the agents in the workspace. Agent $1$ transits from $\pi_2$ to $\pi_1$, agent $2$ from $\pi_1$ to $\pi_2$, and agent $3$ from $\pi_{3}$ to $\pi_2$.
\end{figure*}
\end{proof}


\begin{remark}
Note that the design of the obstacle functions \eqref{eq:betas} renders the control laws \eqref{eq:control law} decentralized, in the sense that each agent uses only local information with respect to its neighboring agents, according to its limited sensing radius. Each agent can obtain the necessary information to cancel the term $\sum_{j\in\mathcal{N}_i(q_i)} \nabla_{q_i}\varphi_j(q)$ from its neighboring agents.  
Finally, note that the considered dynamic model \eqref{eq:dynamics} applies for more general manipulation robots (e.g. underwater or aerial manipulators), without limiting the proposed methodology to mobile ones.

\end{remark}

\subsection{Hybrid Control Framework} \label{subsec:hybrid}
Due to the proposed continuous control protocol of Section \ref{subsec:Continuous design}, the transitions $(\pi_{k_i},t_0)\xrightarrow{i}(\pi_{k'_i},t_{f_i})$ of Problem \ref{Problem} are well-defined, according to Def. \ref{def:agent transition}. Moreover, since all the agents $i\in\mathcal{V}$ remain connected with the subset of their initial neighbors $\widetilde{\mathcal{V}}_i$ and there exist finite constants $t_{f_i}$, such that $\mathcal{A}_i(q_i(t_{f_i}))\in\pi_{k'_i},\forall i\in\mathcal{V}$, all the agents are aware of their neighbors state, when a transition is performed. Hence, the transition system \eqref{eq:TS} is well defined, $\forall i\in\mathcal{V}$.    
Consider, therefore, that $\mathcal{A}_i(q_i(0))\in\pi_{k_{i,0}}, k_{i,0}\in\mathcal{K}, \forall i\in\mathcal{V}$, as well as a given desired path for each agent, that does not violate the connectivity condition of Problem \ref{Problem}. Then, the iterative application of the control protocol \eqref{eq:control law} for each transition of the desired path of agent $i$  guarantees the successful execution of the desired paths, with all the closed loop signals being bounded. 

\begin{remark}
Note that, according to the aforementioned analysis, we implicitly assume that the agents start executing their respective transitions at the same time (we do not take into account individual control jumps in the Lyapunov analysis, i.e., it is valid only for one transition). Intuition suggests that if the regions of interest are sufficiently far from each other, then the agents will be able to perform the sequence of their transitions independently.     
Detailed technical analysis of such cases is part of our future goals.
\end{remark}

\section{SIMULATION RESULTS}\label{sec:simulations}
To demonstrate the validity of the proposed methodology, we consider the simplified example of three agents in a workspace with $r_0 = 12$ and three regions of interest, with $r_k = 4, \forall k\in\{1,2,3\}$ m. Each agent consists of a mobile base and a rigid link connected with a rotational joint, with $\bar{d}_{\scriptscriptstyle B_i} = 1$m, $\forall i\in\{1,2,3\}$. We also choose $p_1 = [-5,-5]$m, $p_2 = [6,-4]$m, $p_3 = [-3,6]$m. 
The initial base positions are taken as $p_{\scriptscriptstyle B_1} = [-3,-4]^\top\text{m}, p_{\scriptscriptstyle B_2} = [3,-4]^\top\text{m}, p_{\scriptscriptstyle B_3} = [-4,-5]^\top\text{m}$ with $\bar{d}_{\scriptscriptstyle B_i} = 1.25\text{m}, \forall i\in\{1,2,3\}$, which imply that $\mathcal{A}_1(q_1(0)),\mathcal{A}_3(q_3(0))\in\pi_1$ and $\mathcal{A}_2(q_2(0))\in\pi_2$ (see Fig. \ref{fig:navigation initial}). The control inputs for the agents are the $2$D force acting on the mobile base, and the joint torque of the link. We also consider a sensing radius of $d_{\text{con}_i} = 8\text{m}$ and the subsets of initial neighbors as $\widetilde{\mathcal{N}}_1 = \{2\}, \widetilde{\mathcal{N}}_2 = \{1,3\}$, and $\widetilde{\mathcal{N}}_3 = \{2\}$, i.e., agent $1$ has to stay connected with agent $2$, agent $2$ has to stay connected with agents $1$ and $3$ and agent $3$ has to stay connected with agent $2$. The agents are required to perform two transitions. Regarding the first transition, we choose $\pi_{k'_1} = \pi_2$ for agent $1, \pi_{k'_2} = \pi_1$ for agent $2$, and $\pi_{k'_3} = \pi_3$, for agent $3$. Regarding the second transition, we choose $\pi_{k'_1} = \pi_1, \pi_{k'_2} = \pi_2$, and $\pi_{k'_3} = \pi_2$. The control parameters and gains where chosen as $k_i = 5, \lambda_i =10, \rho_i=1$, and $\sigma_i = 0.01, \forall i\in\{1,2,3\}$. We employed the potential field from \cite{dimarogonas2007decentralized}.
The simulation results are depicted in Fig. \ref{fig:navigation final}-\ref{fig:c tilde}. In particular, Fig. \ref{fig:navigation final} and \ref{fig:navigation final2} illustrate the two consecutive transitions of the agents. Fig. \ref{fig:gamma beta} depicts the obstacle functions $\beta_i$ which are strictly positive, $\forall i\in\{1,2,3\}$. Finally, the control inputs are given in Fig. \ref{fig:inputs} and the parameter errors $\widetilde{c}$ are shown in Fig. \ref{fig:c tilde}, which indicates their boundedness. As proven in the theoretical analysis, the transitions are successfully performed while satisfying all the desired specifications.

\begin{figure}[]
\vspace{0.4cm}
\centering
\includegraphics[scale=0.35]{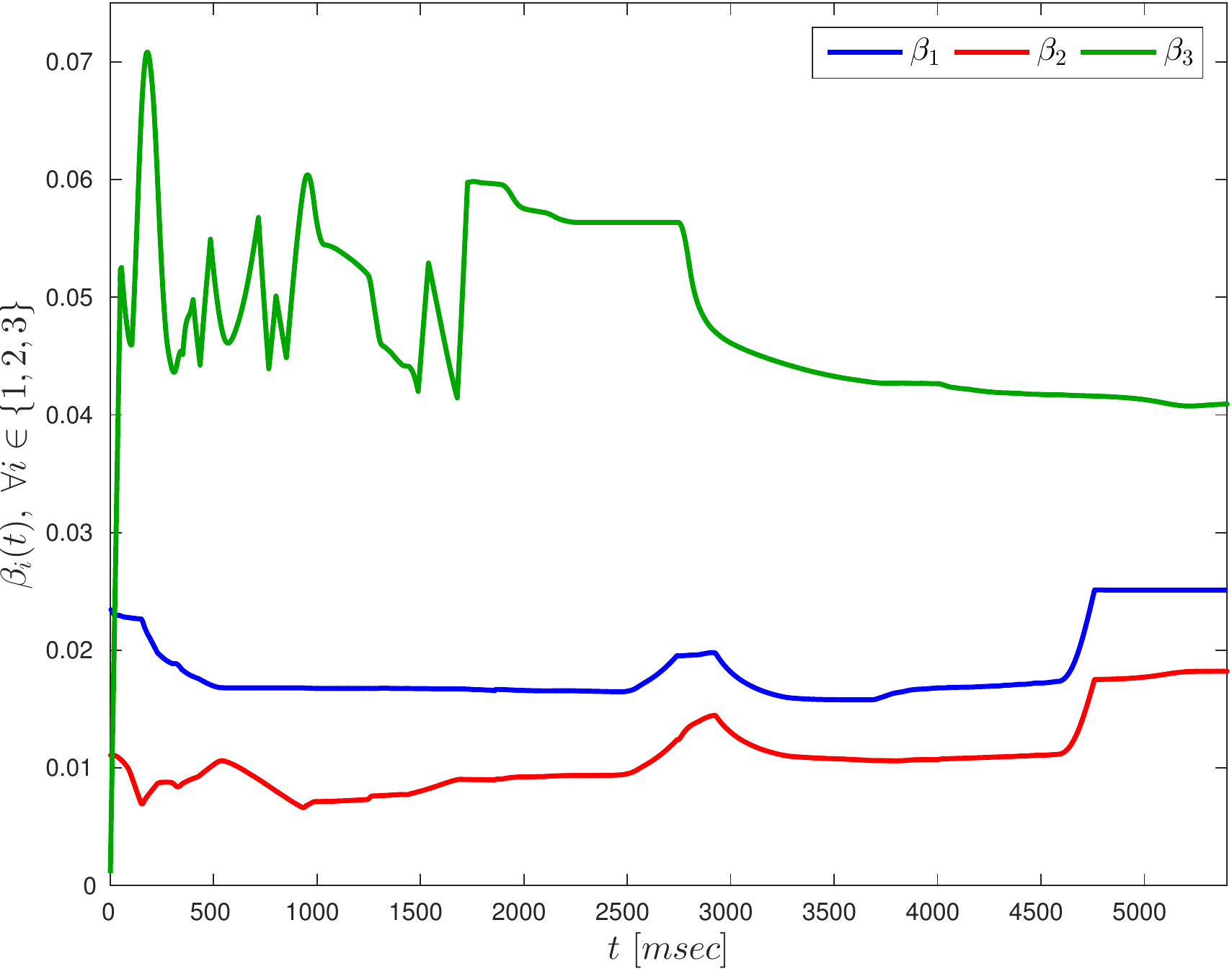}
\caption{The obstacle functions $\beta_i, i\in\{1,2,3\}$, which remain strictly positive.  \label{fig:gamma beta}}
\end{figure}

\begin{figure}[]
\centering
\includegraphics[scale=0.35]{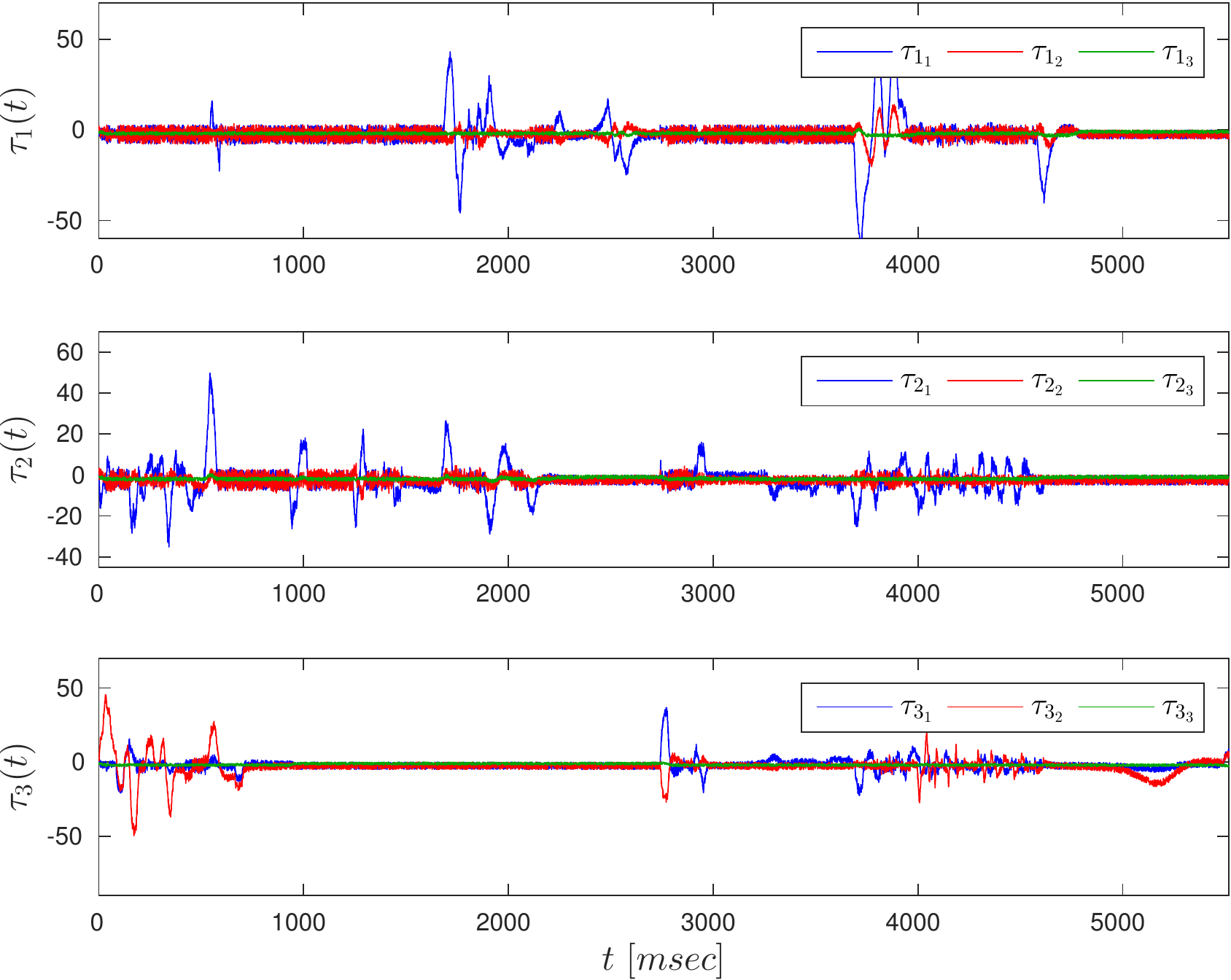}
\caption{The resulting control inputs $\tau_i,\forall i\in\{1,2,3\}$ for the two transitions.\label{fig:inputs}}
\end{figure}

%

\section{CONCLUSIONS AND FUTURE WORKS} \label{sec:concl}
In this paper we designed decentralized abstractions for multiple mobile manipulators by guaranteeing the navigation of the agents among predefined regions of interest, while guaranteeing inter-agent collision avoidance and connectivity maintenance for the initially connected agents. We proposed a novel approach for ellipsoid collision avoidance as well as appropriately chosen potential functions that are free of undesired local minima. Future efforts will be devoted towards addressing abstractions of cooperative tasks between the agents by employing hybrid control techniques as well as abstraction reconfiguration due to potential execution incapability of the transitions.

\begin{figure}[]
\vspace{0.4cm}
\centering
\includegraphics[scale=0.35]{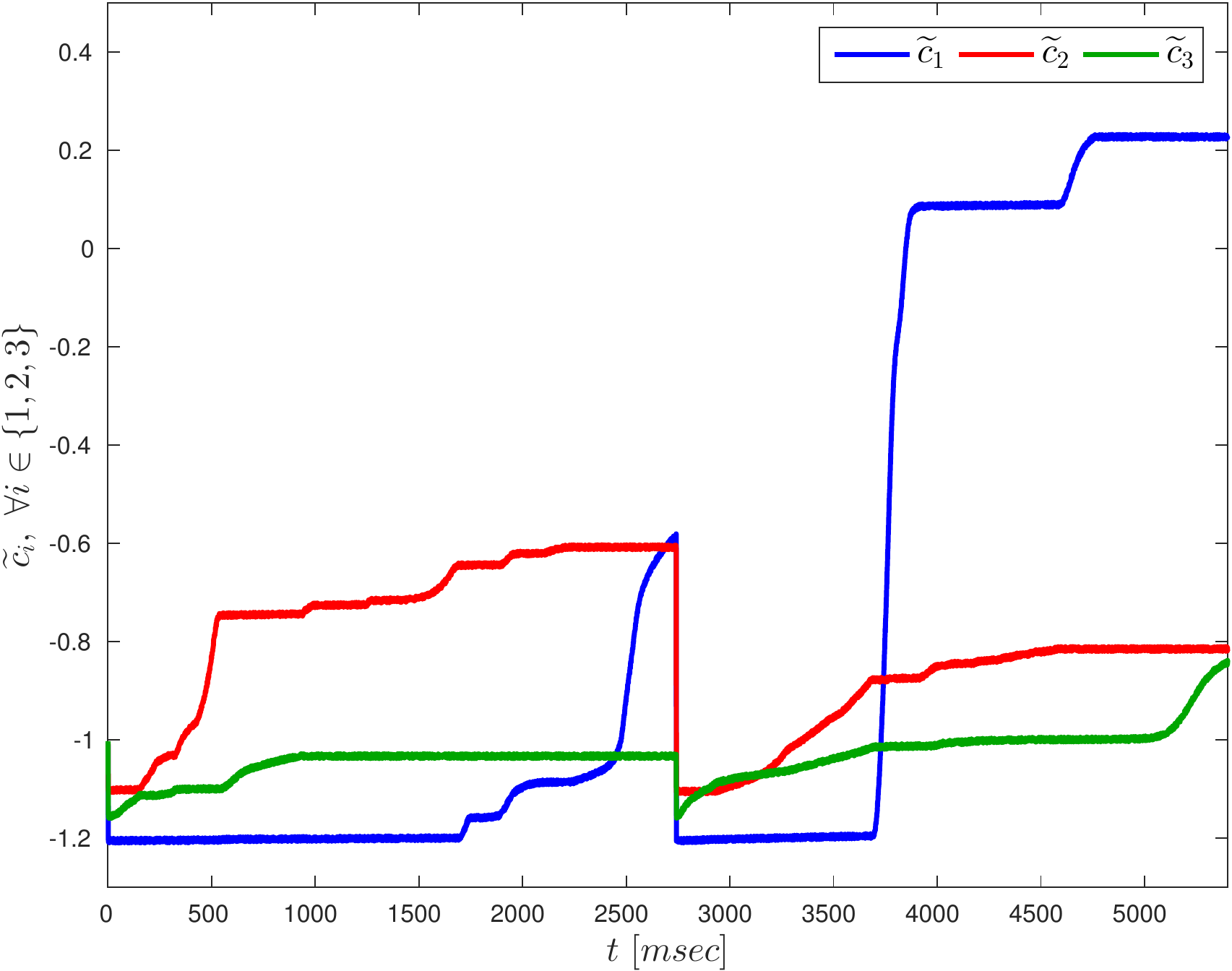}
\caption{The parameter deviations $\widetilde{c}_i,\forall i\in\{1,2,3\}$, which are shown to be bounded.\label{fig:c tilde}}
\end{figure}



\bibliography{references}
\bibliographystyle{ieeetr}

\end{document}